%% file: arxiv-hg-distances.tex
\documentclass{scrartcl}
\usepackage{graphicx}
\usepackage[noblocks]{authblk}
	\usepackage{nicefrac}
	\usepackage{algorithm}
	\usepackage[noend]{algpseudocode}
	\usepackage{enumitem}
	\usepackage{hyperref}
	\usepackage[natbib,maxbibnames=99,style=authoryear,backend=bibtex]{biblatex} 
\usepackage{todonotes}
\usepackage[centertags]{amsmath}
\usepackage{amssymb,amsfonts,amsthm}

\newcommand{\todolater}[1]{}
\newcommand{\todoomit}[1]{}

\newcommand{\ballot}[3]{\left({#1}\,\middle \vert\,{#2}\,\middle \vert\, {#3}\right)}

\newcommand{\rpref}[2]{{\ensuremath{{\overleftarrow{\unrhd_{#2}}({#1})}}}}
\newcommand{\hgopd}{hedonic game with distance-based preferences} 
\newcommand{\hgopds}{hedonic games with distance-based preferences} 
\newcommand{\chgopds}{Hedonic games with distance-based preferences} 
\newcommand{\prefi}[1]{\unrhd_{#1}}

\newtheorem{theorem}{Theorem}
	\newtheorem{proposition}[theorem]{Proposition}

	\newtheorem{definition}[theorem]{Definition}
	\newtheorem{example}[theorem]{Example}
\newtheorem{observation}[theorem]{Observation}

\begin{document}
\title{FEN-Hedonic Games with Distance-Based Preferences}
%
%
\author{Anja Rey}\affil{Universit{\"a}t zu K{\"o}ln, Germany\\ \href{mailto:rey@cs.uni-koeln.de}{rey@cs.uni-koeln.de} \href{https://orcid.org/0000-0001-7387-6853}{https://orcid.org/0000-0001-7387-6853}}
\author{Lisa Rey\thanks{This work is partially supported by DFG-grant BA6270/1-1. We thank anonymous reviewers for helpful comments.}}\affil{Heinrich-Heine-Universit{\"a}t D{\"u}sseldorf, Germany\\ \href{mailto:lisa.rey@hhu.de}{lisa.rey@hhu.de}}
%
\date{}
\maketitle              
\begin{abstract}

Hedonic games formalize coalition formation scenarios where players evaluate an outcome based on the coalition they are contained in. Due to a large number of possible coalitions, compact representations of these games are crucial.
We complement known compact representation models by a distance-based approach:
Players' preferences are encoded in a bipolar manner by ordinal preferences over a small set of known neighbouring players, coalitions are represented by adequate preference orders from a player's perspective, and preferences over coalitions are extended based on a directed form of Hausdorff-Kendall-tau distance between individual preferences and coalitions.
We show that this model satisfies desirable axiomatic properties
and has reasonable computational complexity in terms of selected individual-based stability notions.

\end{abstract}
%
%
%
\section{Introduction}
When basing the decision about whom players share a task with on these players' preferences, a well-established model are hedonic games~\citep{ban-kon-soen:j:core-simple-coalition-formation-game,bog-jac:j:stability-hedonic-coalition-structures}.
\todoomit{
, formally introduced by
\citet{ban-kon-soen:j:core-simple-coalition-formation-game} and
\citet{bog-jac:j:stability-hedonic-coalition-structures}.
}%
As a form of a coalition formation game, players aim to partition into a so-called coalition structure. 
In a hedonic game the individual preferences over such partitions only depend on 
coalitions, i.e. sets of players, they belong to.
The
key
idea of this paper is to model a distance-based representation of hedonic games and study it from an axiomatic and computational point of view.

In general, the number of possible coalitions is exponential in the number of players. Hence, from an algorithmic point of view, it is relevant to find reasonable preference representations that are succinct but
also
expressive~\citep{cha-elk-woo:b:ccgt}.
In very large settings, especially, it is reasonable to assume that for each player the number of known players is small, e.g., bounded by a constant~\citep{pet:graphical-hg-bounded-tw:16,fic-kri-rey:c:hgpt:2019}.
\todolater{noch weitere neuere quellen?}%
Furthermore,
we assume a subdivision of these known co-players into those appreciated (\emph{friends}) and those disapproved of (\emph{enemies})
(see, e.g., \cite{dim-etal:hg-fe:06}).
A player prefers to share a coalition with their friends rather than their enemies.
\todoomit{
We do not, however, stop at counting the number of friends or enemies within a coalition but
}%
Additionally, we
follow the encoding proposed by~
\citet{lan-rey-rot-sch-sch:c:hgopt} in which each player ranks their friends and enemies, respectively. Their model of how
a player compares two coalitions
is based on the polarized responsive extension principle which we are going to discuss in Section~\ref{ch_axiomatic_analysis}.
This encoding allows the expression of a variety of opinions while at the same time players do not have to specify a detailed information, such as a numerical evaluation or knowledge of the whole player set.

The model at hand calculates the distance between a player's preference and a coalition which expresses the player's dissatisfaction with this coalition. While distance-based preferences play an important role in decision making,
preference extensions in hedonic games have not,
to the best of our knowledge, so far, been
based on ordinal distances.
%
We
attend to
the well-studied Kendall-tau distance
which
counts
the minimal necessary swaps to transform one preference into another. In order to allow for indifferences in a player's preference we
consider the Hausdorff-Kendall-tau distance. To this end a
comparable encoding of a preference and a coalition is necessary.
%
%
We show that this model meets desired properties
and compare it with other known hedonic game representations in terms of expressibility and computational complexity of stability problems.

\subsection{Related Work}
Hedonic games have been introduced by
\citet{ban-kon-soen:j:core-simple-coalition-formation-game} and
\citet{bog-jac:j:stability-hedonic-coalition-structures}
and
studied from an axiomatic and algorithmic point of view since (see, e.g.~\cite{cha-elk-woo:b:ccgt,azi-sav:b:handbook-hg}).
\todoomit{woe:c:hg-core:13
}
Applications can be found in particular instances, such as \emph{stable roommate problems}~\citep{stable-roommates,irv:j:stable-roommates}
and \emph{group activity selection problems}~\citep{DarGro18}.
Several
convenient
representations encode preferences over individual players which are then extended to preferences over coalitions. They include network approaches~\citep{dim-etal:hg-fe:06}, numerical approaches (such as additively separable encodings~\cite{bog-jac:j:stability-hedonic-coalition-structures}) and ordinal approaches (such as the singleton encoding~\cite{cec-rom:j:singleton-encoding}). More recently, encodings have been considered that do not assume that each player knows every other player in the game, but that they only know a subset of players and consider the others as neutral~\citep{ota-etal:hg-fen:17,pet:graphical-hg-bounded-tw:16,lan-rey-rot-sch-sch:c:hgopt}, sometimes referred to as FEN-hedonic games.
Hedonic games with ordinal preferences and thresholds~\citep{lan-rey-rot-sch-sch:c:hgopt,ker-lan-rey-rot-sch-sch:j:hgopt} combine trichotomous preferences with ordinal preferences.
In this paper we refer to this encoding.
So far, the extensions to coalitions comprise a set of possible and necessary extensions
as well as a numerical approach involving Borda scores~\citep{lan-rey-rot-sch-sch:c:hgopt}.
The additional requirement of a constant number of known players enables computational lower bounds to decrease and we cannot transfer proof techniques to our model immediately.
\citet{pet-elk:c-coopmas:simple-causes-complexity-hg} study causes of complexity for hedonic games. Unfortunately, the results are not applicable here.
We do, however, fulfil the requirements of a graphical hedonic game~\citep{pet:graphical-hg-bounded-tw:16}.
Consequently, several problems regarding stability of a game are fixed-parameter tractable with respect to treewidth and degree of the underlying dependency graph.
Note that social distance games~\citep{social-distance-games}
and distance hedonic games~\citep{fla-kod-ols-var:c:distance-hedonic-games}
are a different approach
concerning distances of players in a graph.
\todolater{neue aktuelle Arbeiten?}%

Preference extensions from ordinal preferences over single items to a preference order over subsets of items (see, e.g.~\cite{bar-bos-pat:b:ranking-sets-of-objects}) find applications in various topics within the field of \emph{Computational Social Choice}~\citep{bra-con-end-lan-pro:b:handbook-comsoc}.
Distance-based approaches in order to define these extensions are suggested in several contexts, but have not, to the best of our knowledge, yet been modelled for hedonic games.
%
Closely related to our research is the field of \emph{committee elections} (see, e.g.,~\cite{fal-sko-sli-tal:b:multiwinner-voting-trends-comsoc}) in which a subset of candidates, the so-called committee, is
elected
based on voters' preferences.
In this context
\citet{bra-kil-san:j:minimax-electing-committees} proposed to use distances
as closeness measures and established a minimax approach, i.e., the election winner minimizes the maximal distance between preferences and committee. As in our work, this method includes the necessity of representing the committee (as in our case the coalition) as well as the preference in the same form. Their approach, however, regards
\todoomit{
approval votes and thus preferences 
}%
votes
which solely express binary opinions (approval or disapproval) of candidates. Thereby, it suffices to make use of the Hamming distance.
This approach was extended to 
other forms of preferences, including linear orders
\citep{bau-den:c:voter-dissatisfaction-committee-elections}.
%
%
Similarly, in \emph{judgment aggregation} (see, e.g.~\cite{end:b:handbook-ja}),
aiming at
a collective judgment set out of individual judgments, distance-based procedures
measure the Hamming distance of a complete and consistent judgment set to the individual judgments and return the one minimizing the sum of
distances. 
\todoomit{
In \emph{resource allocation}~(see, e.g.~\cite{bou-che-mau:b:handbook-resource-allocation}) the extension of
ordinal
preferences over objects to preferences over subsets of objects also play a role.
}%

We refer to a number of axiomatic properties which are desirable for a hedonic game. These are in parts inspired by different fields of Computational Social Choice as described above and preference-based matching.
Fairness properties such as \emph{anonymity} and \emph{nonimposition} stem from a voting context.
\emph{Responsiveness} as introduced by
\citet{roth:j:responsive} in the context of many-to-one matching markets refers to a guaranteed benefit from replacing a former match by a match with a preferred item.
Basic properties for extensions to subsets of any size are defined by
\citet{bar-bos-pat:b:ranking-sets-of-objects} and
\citet{bos-sch:j:minimal-paths-on-ordered-graphs}, whose extension principle is 
defined for
polarized preferences by
\citet{lan-rey-rot-sch-sch:c:hgopt}.
\todoomit{Our extension is compatible with this polarized responsiveness extension principle.
}%



%
%
%





\subsection{Contribution}
We define the model of \emph{\hgopds} in Section~\ref{ch_model}:
In order to extend a player's opinion encoded by polarized ordinal preferences
over a subset of known players (3.1) to a preference order over coalitions which constitute the game (3.4),
we represent a coalition also as a preference order from the player's perspective (3.2), and compute a distance between these two preference orders (3.3).
We support this definition by an easily accessible characterization.
In Section~\ref{ch_background} we discuss why this particular coalition representation and distance
outperforms
other approaches.
In Section~\ref{ch_axiomatic_analysis} we analyse our model axiomatically: Desired properties
are fulfilled,
for instance, that a coalition cannot become less favourable when a preferred player enters the coalition.
We
observe this distance-based approach to be a reasonable and natural completion of existing compact representations which is at the same time flexible enough in terms of expressivity.
We study computational aspects of individual-based stability in Section~\ref{ch_stability}
and show that, e.g., Nash stability of a given outcome can be verified in linear time while it is NP-complete to decide whether there exists a Nash-stable outcome in a given game.

\section{Preliminaries}
Our work is based on the game-theoretic study of hedonic games as well as on distance measures which are outlined in the following two subsections.

\subsection{Hedonic Games}
We
consider
a class of cooperative games, where players want to
partition into subsets of players, so-called \emph{coalitions}.
A \emph{hedonic game}~\citep{bog-jac:j:stability-hedonic-coalition-structures,ban-kon-soen:j:core-simple-coalition-formation-game} is such a coalition formation game~$\langle N,\succeq\rangle$ with
a set of $n=\lvert N \rvert$ players
and a preference profile $\succeq=(\succeq_1,\dots,\succeq_n)$ such that each player~$i$'s preference $\succeq_i$ only depends on the coalitions $i$ is contained in.
The set of such coalitions is denoted by~$\mathcal{N}_i=\{C\subseteq N\mid i\in C\}$.
A partition $\Gamma$ of the players into subsets is called a \emph{coalition structure} and $\Gamma(i)$ denotes the coalition in $\Gamma$ containing player $i$. 

A common
goal of these games is to achieve stable coalition structures
no player or group of players has an intention to deviate from.
A coalition structure~$\Gamma$
is called \emph{Nash stable}, if there is no player~$i\in N$ who prefers another coalition~$C\cup\{i\}$ with $C\in \Gamma$ to $\Gamma(i)$.
If a deviation is restricted to coalitions where $i$ is welcome ($C\cup\{i\}\succeq_j C$ for each $j\in C$) and if there is no such deviation, $\Gamma$ is called \emph{individually stable}.
If additionally $i$ can only deviate if not bound to a contract in $i$'s former coalition ($\Gamma(i)\setminus\{i\}\succeq_j \Gamma(i)$ for each $j\in \Gamma(i)\setminus\{i\}$),
$\Gamma$ is called \emph{contractually individually stable}.
Moreover, $\Gamma$ is called \emph{perfect} if each $\Gamma(i)$ is one of player~$i$'s favourite coalitions ($\Gamma(i)\succeq_i C$ for each $C\in\mathcal{N}_i$).
%
For an analysis of stability, usually, two decision problems are distuinguished: the verification problem of whether a given coalition structure is stable, and the existence problem of whether a game allows a stable coalition structure.

In order to achieve succictness,
often a player's opinion is encoded as a preference order~$\unrhd$ over the set of players
which can be extended to a preference order~$\succeq$ over coalitions.
We consider $\unrhd$ to be an ordinal preference order with possible indifferences, i.e., a linear order (reflexive, transitive), that is not necessarily antisymmetric.
For two players~$x,y\in N$, $x$ is
\emph{weakly preferred to} $y$ if $x\unrhd y$;
$x$ is \emph{(strictly) preferred to} $y$ if $x\rhd y$ (i.e., $x\unrhd y$ and not $y\unrhd x$);
and $x$ and $y$ are considered \emph{indifferent} if $x\sim y$ (i.e., $x\unrhd y$ and $y\unrhd x$).
Similarly, for
$\succeq$,
we have
$A\succ B$ if and only if $A\succeq B$ and not $B\succeq A$.
By $\unrhd(C)$ we denote the restriction of $\unrhd$ to the players within $C\subseteq N$.
Let the \emph{rank} of a player~$j$ be the position among equally liked players within a preference order: $j$ is among the top-ranked players (of rank~$1$) if $j\unrhd k$ for each $k\in N$; $j$ is of rank~$p$ if
there exist $p-1$ other players $k_1,\dots,k_{p-1}$ such that $k_1\rhd\dots\rhd k_{p-1}\rhd j$, but no $p$ players $\ell_1,\dots,\ell_{p}$ such that $\ell_1\rhd\dots\rhd \ell_p\rhd j$.
Due to indifferences, several players can share the same rank, such that the maximal rank
is possibly smaller than $n$.

\subsection{Distance Measures}


We study
the distance between two preference orders, i.e., linear orders over the set of players. To be precise, in our setting a player's opinion is compared to a coalition, the former given in form of a 
  preference order over a subset of players;
  the latter interpreted as a preference order 
  from the player's perspective.
Before we go into detail in Section~\ref{ch_model}, we present background on
how to compare
preference orders in general.
%
 A \emph{distance measure} on a space $A$ is a metric $\mathrm{dist}: A
    \times A \rightarrow \mathbb{R}_{\geq 0}$, i.e.
    for each $a,b,c\in A$ it fulfils non-negativity ($\mathrm{dist}(a,b) \geq 0$), identity of indiscernibles ($\mathrm{dist}(a,b)=0 \iff a=b$), symmetry ($\mathrm{dist}(a,b)=\mathrm{dist}(b,a)$), and the triangle inequality ($\mathrm{dist}(a,b) + \mathrm{dist}(b,c) \geq \mathrm{dist}(a,c)$).
 A directed distance is a quasimetric which does not fulfil symmetry, or a pseudoquasimetric, if it additionally does not fulfil identity of indiscernibles.
    

     
    For strict 
    preference orders, that is, antisymmetric linear orders, a well-known distance metric is
    the Kendall-tau distance. 
   Based on Kendall's measure of rank correlation~\citep{ken:j:a-new-measure-of-rank-correlation}, it
   calculates the minimal number of inversions of adjacent players necessary to convert one
     strict preference order~$\rhd^\alpha$ into another, $\rhd^\beta$:\\[-2ex]
     \[\tau\left(\rhd^\alpha,\rhd^\beta\right) = \lvert\{(x,y) \in N \times N\ \mid \ x \rhd^{\alpha} y \textnormal{ and } y \rhd^{\beta} x\} \rvert.\]
     
    
     
    In our context, preference orders allow indifferences (i.e., are not necessarily antisymmetric).
    Any such order~$\unrhd$ can be interpreted as an equivalence class containing all strict preferences orders~$\rhd^\alpha$ which are consistent with $\unrhd$ (i.e., $x\rhd^\alpha y$ if $x\rhd y$)
    and each indifference is resolved by some permutation (either $x\rhd^\alpha y$ or $y\rhd^\alpha x$, if $x\unrhd y$ and $y\unrhd x$).
    In order to include indifferences, we consider a Hausdorff distance~\citep{hau:b:mengenlehre}.
The Hausdoff--Kendall-tau distance is defined by\\[-1ex]
    \[\tau^*(\unrhd^\pi,\unrhd^\sigma) = \max\left\{ \max_{\rhd^\beta \in \unrhd^\sigma} \min_{\rhd^\alpha \in \unrhd^\pi} \tau\left(\rhd^\alpha, \rhd^\beta\right), \max_{\rhd^\alpha \in \unrhd^\pi} \min_{\beta \in \unrhd^\sigma} \tau\left(\rhd^\alpha, \rhd^\beta\right)\right\}\]
    for two preference orders $\unrhd^\pi$ and $\unrhd^\sigma$ interpreted as equivalence classes of strict preferences
    with maximal ranks~$r_\pi$ and~$r_\sigma$, respectively.
    %
\todoomit{
    This can be characterized (see, e.g.,\cite{cri:b:metric-ranked-data}) by
    \[\tau^*(\unrhd^\pi,\unrhd^\sigma) = \max\left\{
     \sum_{\substack{p,p'\in\{1,\dots,r_{\pi}\}\\q,q'\in\{1,\dots,r_{\sigma}\}\\p<p', q\geq q'}}n_{pq}\cdot n_{p'q'},
     \sum_{\substack{p,p'\in\{1,\dots,r_{\pi}\}\\q,q'\in\{1,\dots,r_{\sigma}\}\\p\leq p', q> q'}}n_{pq}\cdot n_{p'q'}
     \right\}
    \]
}%
%
    Intuitively, for the first part,
    a \emph{worst-case} resolution~$\rhd^\beta$ of $\unrhd^\sigma$ is chosen for which the 
    the minimal number of swaps between some $\rhd^\alpha$ resolving $\pi$ and $\rhd^\beta$ is
    maximized.
  %
  %
    For our interpretation of a coalition depending on a player's perspective, we only consider this first part; see Section~\ref{ch_model} for details.
    We denote this
    directed distance
    by
    \begin{align}
      \overrightarrow{\tau}\left(\unrhd^\pi,\unrhd^\sigma\right) = \max_{\rhd^\beta \in \unrhd^\sigma} \min_{\rhd^\alpha \in \unrhd^\pi} \tau\left(\rhd^\alpha, \rhd^\beta\right)
      =\sum_{\substack{p,p'\in\{1,\dots,r_{\pi}\}, p<p'\\
      q,q'\in\{1,\dots,r_{\sigma}\}, q\geq q'
      }}n_{pq}\cdot n_{p'q'},
      \label{eq:char-sum}
    \end{align}
    where the latter is a characterization
    (see, e.g.,\cite{cri:b:metric-ranked-data})
    with
    the number $n_{pq}$ of items that are
    ranked on position $p$ in $\unrhd^{\pi}$
    and on position $q$ in $\unrhd^{\sigma}$.
    While in general $\overrightarrow{\tau}$ is a pseudoquasimetric,
    the directed distance
    as defined in Section~\ref{subsec:delta} allows a quasimetric. See Section~\ref{ch_axiomatic_analysis} for details.
    This matches the Hausdorff-Kendall-tau metric via
 %
  $  \tau^*(\unrhd^\pi,\unrhd^\sigma) = \max\{\overrightarrow{\tau}\left(\unrhd^\pi,\unrhd^\sigma\right),\overrightarrow{\tau}\left(\unrhd^\sigma,\unrhd^\pi\right) \}$.




\section{The Model}
\label{ch_model}
This section presents a compactly encoded model of hedonic games which yet gives credit to each player's detailed opinion on a subset of candidates. The preference extension is based on calculating the Kendall-tau distances between a player's preference and  coalitions containing this player. Thus, a player compares two coalitions by ranking the coalition with minimal distance, i.e., the minimal dissatisfaction, highest. 
%
We create a \emph{\hgopd} as follows:
Given a player's ordinal preferences over the
known players,
we define
\begin{itemize}\setlength\itemindent{-1em}
    \item[](3.2) a preference-based encoding of a coalition from a player's perspective and
    \item[](3.3) a distance between preference and coalition from a player's perspective
\end{itemize}
which determines the extension to preferences over coalitions for each player.

\subsection{Preference Encoding}

  We define the preference encoding similar to 
  \emph{hedonic games with ordinal preferences and 
  thresholds}~\citep{lan-rey-rot-sch-sch:c:hgopt}. For each player~$i$ the set of other players is partitioned into \emph{accepted} players~$N^+_i$ (\emph{friends}),
  \emph{unaccepted} players~$N^-_i$ (\emph{enemies}),
  and neutral (e.g. unknown) players~$N^0_i$.
    This can be represented by a directed graph, the underlying \emph{dependency graph} with two edge labels $+$ and $-$.
    By $N_i=N^+_i\cup N^-_i$ we denote $i$'s neighbourhood
    and
    assume that $\lvert N_i\rvert=\lvert N^+_i\cup N^-_i\rvert\leq d$ for a constant $d>0$. 
    Additionally, each player~$i$ specifies a partial preference relation $\unrhd_i(N^+_i)$ over $N^+_i$ and
    a partial preference relation $\unrhd_i(N^-_i)$ over $N^-_i$ such that a ballot contains information of 
    the form $\ballot{\unrhd_i(N^+_i)}{(N^0_i \cup \{i\})_{\sim}}{\unrhd_i(N^-_i)}$, where $S_{\sim}$ for any $S \subseteq N$ denotes that all players in $S$ are 
    ranked equally.
    This induces a preference relation $\unrhd_i$: For players $j,k\in N\setminus\{i\}$ it holds that $j\unrhd_i k$ if and only if
   $j,k\in N^+_i$ and $j\unrhd_i k$; $j\in N^+_i$ and $k\in N^0_i$; $j\in N^+_i$ and $k\in N^-_i$; $j\in N^0_i$ and $k\in N^-_i$; or $j,k\in N^-_i$ and $j\unrhd_i k$.
   In order to obtain our final preference encoding we 
   conduct two steps. Firstly, since the neutral players in 
   $N^0_i$ do not have any effect on the evaluation 
   of a coalition we abbreviate the ballot notation by
   $\ballot{\unrhd_i(N^+_i)}{i}{\unrhd_i(N^-_i)}$. Secondly, in order to facilitate a comparison with a coalition, we subdivide the ballot into a part regarding player $i$'s friends
\todoomit{   
   \[\unrhd^+_i = \ballot{\unrhd_i(N^+_i)}{i}{}\] and 
   $i$'s enemies 
   \[\unrhd^-_i = \ballot{}{i}{\unrhd_i(N^+_i)}.\]
}
and enemies, respectively,\\[-2ex]
  \[
   \unrhd^+_i = \ballot{\unrhd_i(N^+_i)}{i}{ },\quad
   \unrhd^-_i = \ballot{}{i}{\unrhd_i(N^+_i)}.
  \]
  This again induces preference orders, namely $\unrhd^+_i=\unrhd_i(N^+_i\cup\{i\})$ and $\unrhd^-_i=\unrhd_i(N^-_i\cup\{i\})$ in which $j \rhd^+_i i$  for all $j \in N^+_i$ ($i \rhd^-_i j$ for all
   $j \in N^-_i$) complements $\unrhd_i(N^+_i)$ ($\unrhd_i(N^-_i)$, respectively).
%
  \begin{example}\label{ex:running-example-1}
   For instance, let $i$ know five other players, $a$, $b$, $c$, $d$, and $e$, three of which $i$ likes ($a$, $b$, $c$) and two of which $i$ doesn't like ($d$, $e$), and let $i$ specify preferences with whom to cooperate by
   $\ballot{a\rhd_i b\sim_i c}{i}{d\rhd_i e}$.   
  \end{example}

\subsection{Coalition Encoding}
  In order to compare a player $i$'s ballot with a coalition $C \in \mathcal{N}_i$, we interpret $C$ from 
  $i$'s perspective.
  Intuitively, 
  friends in $C$ and absent enemies are ranked according to 
  $i$'s opinion, whereas missing good friends and present unfavourable enemies have a reversed impact.
  By \rpref{S}{i}  we denote player $i$'s reversed preference order over all players in 
  $S \subseteq N$ such that the separate representation is defined matching the form of two ballots\\[-3.4ex]
  \begin{align*}
  C^+_i &= \ballot{\unrhd_i(C \cap N^+_i)}{i}{\rpref{N^+_i \setminus C}{i}},\quad
  C^{-}_{i} =  \ballot{\rpref{C \cap N^-_i}{i}}{i}
  {\unrhd_i(N^-_i \setminus C)}.
\end{align*}
Note that this is again the abbreviated notion, omitting
neutral players which are implicitly ranked equally 
to player~$i$ without influencing player~$i$'s preference order 
over coalitions. By this
we allow 
a partial ballot to
be compared to
a partial coalition
containing the identical subset of players.

\begin{example}\label{ex:running-example-2}
  Continuing Example~\ref{ex:running-example-1}, from $i$'s perspective, we can now encode coalitions contained in $\mathcal{N}_i$.
  For a coalition $C=\{i,a,b,e\}$ (and any combination with further neutral players), we obtain $C_i^+=\ballot{a\rhd_i b}{i}{c}$ and $C_i^-=\ballot{e}{i}{d}$. For $D=\{i,c,d,e\}$, we have $D_i^+=\ballot{c}{i}{b\rhd_i a}$ and $D_i^-=\ballot{e\rhd_i d}{i}{}$. The coalition containing all neighbouring players, $E=\{i,a,b,c,d,e\}$ is encoded by $E_i^+=\ballot{a\rhd_i b\sim_i c}{i}{}$ and $E_i^-=\ballot{e\rhd_i d}{i}{}$; $F=\{i,a,b,c\}$ containing all friends, by $F_i^+=\ballot{a\rhd_i b\sim c}{i}{}$ and $F_i^-=\ballot{}{i}{d \rhd_i e}$.
\end{example}

\subsection{Distance Between Preference and Coalition}
\label{subsec:delta}
Given representations of both
a player's preferences
and a coalition, we calculate their distance by summing up the separate distances regarding all friends and all enemies.
Our model defines this distance between
$\unrhd_{i}$ and
$C\in\mathcal{N}_i$ by
\[\delta(\unrhd_{i}, C) = \delta^+(\unrhd_{i}, C) 
+ \delta^-(\unrhd_{i}, C)\] with
\[\delta^+(\unrhd_{i}, C) = \overrightarrow{\tau}\left(\unrhd^+_{i}, C_i^+\right) \textnormal{ and } \delta^-(\unrhd_{i}, C) = \overrightarrow{\tau}\left(\unrhd^-_{i}, C_i^-\right),\]
for which $\unrhd^+_{i}$ and $C_i^+$ ($\unrhd^-_{i}$ and $C_i^-$) induce, as defined, equivalence classes containing all
antisymmetric linear orders over $N^+_{i} \cup \{i\}$
($N^-_{i} \cup \{i\}$) dissolving occurring indifferences in the ballots.


Note that we consider the directed distance $\overrightarrow{\tau}$ here, since we have a dependence
between the coalition representation and the preferences of a player.
Considering the worst case of possible swaps within indifferences in the preference encoding, would contradict this intuition.
In fact, it always holds that $\overrightarrow{\tau}\left(\unrhd^+_{i}, C_i^+\right)\leq \overrightarrow{\tau}\left( C_i^+,\unrhd^+_{i}\right)$.
%
Furthermore, note that $\delta$ could be any other function in $\delta^+$ and $\delta^-$, for instance the Euclidean distance. If not stated otherwise, we consider the 1-norm in this paper.

\begin{example}\label{ex:running-example-3}
 For the instance in Examples~\ref{ex:running-example-1} and~\ref{ex:running-example-2}, we now obtain the distances from $i$'s preferences to the coalitions by 
 $\delta^+(\unrhd_i,C)=1$, $\delta^-(\unrhd_i,C)=2$, thus, $\delta(\unrhd_i,C)=3$;
 $\delta(\unrhd_i,D)=4+3=7$;
 $\delta(\unrhd_i,E)=0+3=3$; and
 $\delta(\unrhd_i,F)=0$.
\end{example}

We can characterize
this distance as follows.
\begin{proposition}\label{prop:char-delta}
 It holds that
 \begin{align}
   \delta^+(\unrhd_i,C)= \lvert N_i^+\setminus C\rvert + \sum_{f\in N_i^+\setminus C} \left\lvert \left\{b\in N_i^+ \mid f\rhd_i b \right\} \right\rvert
   \label{eq:char-deltaplus}
 \\[-2ex]
 \intertext{and}
   \delta^-(\unrhd_i,C)= \lvert C\cap N_i^-\rvert
   + \sum_{f\in N_i^-} \left\lvert \left\{b\in C\cap N_i^- \mid f\rhd_i b \right\} \right\rvert.
   \label{eq:char-deltaminus}
 \end{align}
\end{proposition}
\begin{figure}[t]
  \centering
   \begin{minipage}{0.45\textwidth}
    \centering
    \def\svgwidth{7.3cm}
    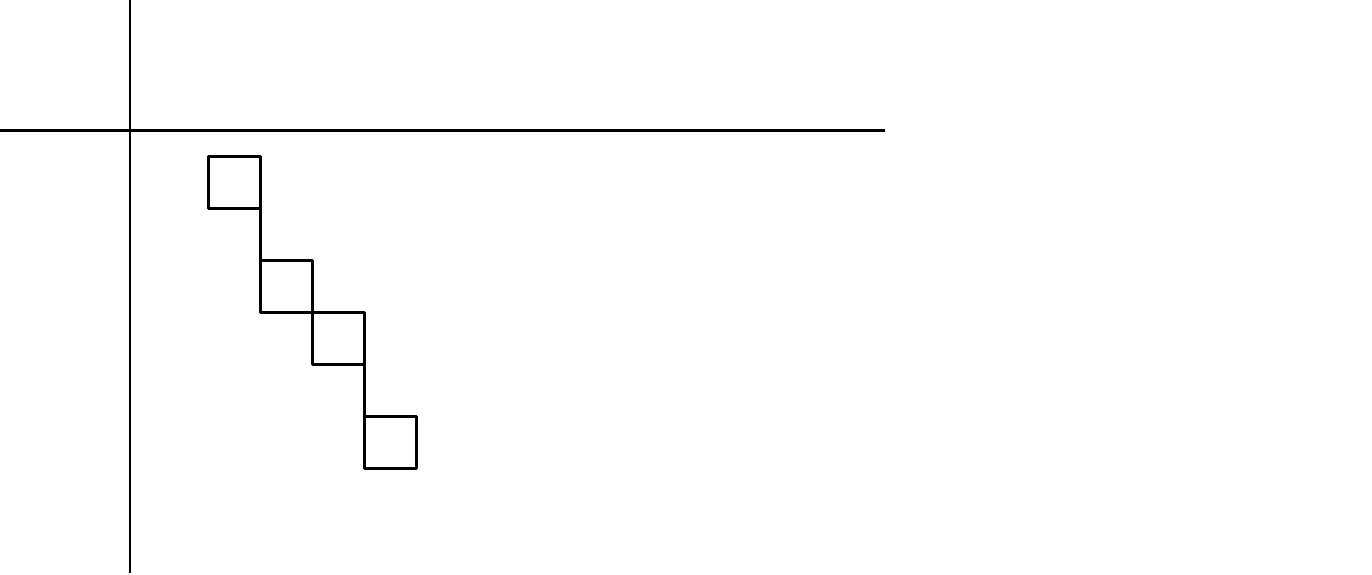
   \end{minipage}
   \begin{minipage}{0.45\textwidth}
    \centering
    \def\svgwidth{5.8cm}
    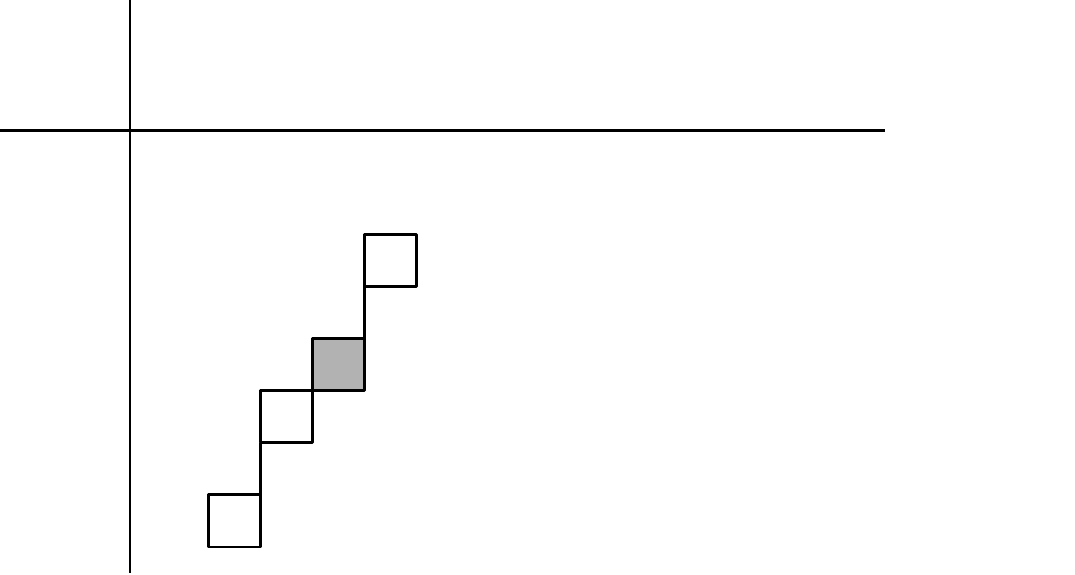
   \end{minipage}    
    \caption{Characterization of $\delta^+(\unrhd_i,C)$ and $\delta^-(\unrhd_i,C)$}
    \label{fig:char-dist}
\end{figure}
%
Figure~\ref{fig:char-dist} depicts
the characterization of $\delta^+$ (Equation~\eqref{eq:char-sum}, see also~\cite{cri:b:metric-ranked-data})
on the left hand side.
Note that each row contains at most two non-zero entries, since all equally ranked players are either positioned in the same rank within or not within the coalition.
Moreover, 
each column only has one non-zero entry, since each player in the same rank in the coalition (and outside the coalition) origins from a common rank in the preference.
In the bottom row, $i$ marks a fixed position in both orders.
We sum up the following parts.\\[-4ex]
\begin{itemize}
    \item For each $f\in C\cap N_i^+$ shifting from rank $p$ in the preference to $q$ in the coalition,  there is no non-zero entry $n_{p'q'}$ with $p'>p$ and $q'\leq q$.
    \item Each $f\in N_i^+\setminus C$ shifting from rank $p$ in the preference to $q$ outside the coalition is multiplied with the entries in the rectangle below and left of $p$ and $q$. These are non-zero for (a) each $b$ shifting from $p'<p$ (i.e., $f$ is preferred to $b$) to any $q'$ (both inside and outside the coalition) and (b) the $i$-marker in the bottom row.
\end{itemize}
This sums up to Equation~\eqref{eq:char-deltaplus}.
Similarly, for $\delta^-$, Equation~\eqref{eq:char-deltaminus} is obtained, see also Figure~\ref{fig:char-dist}, right hand side.

\subsection{Preference Extension}
Now, we make use of this distance notion in order to define our game's preference extension.
A player~$i\in N$ weakly prefers a coalition $A\in \mathcal{N}_i$ to a coalition $B\in \mathcal{N}_i$
if and only if $A$ is at most as far from $i's$ preference order as $B$. Formally,
\[
  A\succeq_i B\iff \delta(\prefi{i},A) \leq \delta(\prefi{i},B).
\]

\begin{example}\label{ex:running-example-4}
 Finally, for player~$i$ in Examples~\ref{ex:running-example-1},~\ref{ex:running-example-2}, and~\ref{ex:running-example-3}, we obtain the following preferences over the example calitions:
 $F$ is $i$'s favourite coalition with distance~$0$.
 Coalition $C$ is preferred to $D$ ($3<7$) which is also reasonable since $C$ contains more and better friends, while $D$ has an additional enemy. Coalition $E$ cannot be compared to $C$ that easily with more friends, but also more enemies. The distance, however, is the same for both $C$ and $E$. Therefore, $C\succeq_i E$ and $E\succeq_i C$.
\end{example}

Hence, with $\succeq=(\succeq_1,\dots,\succeq_n)$ we obtain a
\emph{\hgopd} $\langle N,\succeq\rangle$.
In Section~\ref{ch_axiomatic_analysis} we show that this model satisfies desired axiomatic properties.

\section{The Model's Background}
\label{ch_background}
In this section we argue the steps toward the above defined
model
to achieve our aim of finding a Kendall-tau-based preference extension.
For a Hausdorff--Kendall-tau approach
extending player~$i$'s polarized ordinal preference encoding $\ballot{\unrhd_i(N^+_i)}{i}{\unrhd_i(N^-_i)}$, we rule out the below described alternative representations.
%
\todoomit{A player's satisfaction with a coalition is higher when their most preferred players are members of the coalition while their least preferred players are not.}%
As a starting point
%
we
consider a
representation of coalition~$C$
where players within and outside of $C$ remain unranked,
such as
%
$\ballot{(C \cap N_i)_{\sim}}{i}{(N_i \setminus C)_{\sim}}$.
This way the number of necessary swaps is minimal for player $i$'s best friend and maximal for their least preferred player.
Strongly opposing the advantage of a coalition
independent of the player's preference is the following example which shows that adding a friend to a coalition can make the coalition worse. 

\begin{example}
   \label{ex_model_background_version1}
   Consider player $i$'s preference $\unrhd_{i} = 
   \ballot{a \rhd b \rhd c \rhd f}{i}{}$.
   Coalition $A = \{b,c,i\}$ is
   preferred to coalition $A \cup \{f\}$.
\end{example}

\todoomit{
The property that adding a friend to a coalition should not increase the player's dissatisfaction with this coalition will be discussed in Section~\ref{ch_axiomatic_analysis} and is considered necessary.
}%
Discarding equivalences between all candidates within
and outside
the coalition leads to the adjusted representation $C = \ballot{\unrhd_{i}(C \cap N^+_i)}{i}{\unrhd_{i}(N^-_i \setminus C)}$. Two different problems arise in this setting. Firstly, still, adding a friend to a coalition can increase player $i$'s dissatisfaction
as Example~\ref{ex_model_background_version2} shows. Note that reversing the order of players
outside the coalition does not change this. Secondly, this ordering contains friends and enemies alike, thus weighting friends higher than enemies, a constraint not going in line with the idea of the model.

\begin{example}
   \label{ex_model_background_version2}
  Again consider
  $\unrhd_{i} = 
   \ballot{a \rhd b \rhd c \rhd f}{i}{}$.
   It holds that $\{i\}\succ_i \{i,f\}$.
\end{example}

Resulting from this, we consider distances separately for friends and enemies.
We subdivide the preference ballot $\unrhd_{i}$ in the same manner
in order to compare ballots which contain the exact same set of players.
Keeping the previous coalition representation intact,
a subdivision would lead to the coalition ballots 
$C^+_i = \ballot{\unrhd^+_{i}(C \cap N^+_i)}{i}{\unrhd^+_{i}(N^+_i \setminus C)}$ and $C^-_i = \ballot{\unrhd^-_{i}(C \cap N^-_i)}{i}{\unrhd^-_{i}(N^-_i \setminus C)}$. Example~\ref{ex_model_background_version3} shows that adding a friend again increases the distance.

\begin{example}
    \label{ex_model_background_version3}
   For $\unrhd_{i} = 
   \ballot{a \rhd b \rhd c \rhd f}{i}{}$,
   $i$
   prefers $\{a,i\}$ to $\{a,f,i\}$.
\end{example}

As a last step this postulates that the order of those players who are not in the desired positions is reversed in the coalition, i.e., the order of all friends who are not within the coalition and of all enemies who are within the coalition is reversed.
At this point we work with the representations presented in Section~\ref{ch_model} for which we show
desirable properties
in Section~\ref{ch_axiomatic_analysis}

Apart from the representation of preferences and coalitions we make use of a directed distance measure. Next to the above discussed argument that the coalition is represented in relation to the preference,
an axiomatic analysis underlines the importance of basing the model on a directed version of the 
Hausdorff-Kendall-tau distance. 
Example~\ref{ex_model_background_yellow_sum} shows, that again adding a friend to a coalition can make a coalition worse when considering the undirected distance.

\begin{example}
   \label{ex_model_background_yellow_sum}
   
   Consider player $i$'s preference $\unrhd_{i} = 
\ballot{a \sim_{i} b \sim_{i} c \sim_{i} d \sim_{i} f}{i}{}$. Then, player~$i$
prefers 
coalition $\{a, i\}$
to $\{a,f,i\}$ when calculating the distance via the (undirected) Hausdorff-Kendall-tau distance.
\end{example}





This leads us to both the representation of coalitions and preferences defined in Section~\ref{ch_model} as well as to the use of a directed distance.

\section{Axiomatic Analysis}
\label{ch_axiomatic_analysis}
In the following, we study properties of the directed distance $\delta$ between a player's preference and a coalition as well as the preference extension. Since these properties are from one player~$i$'s perspective, we omit the index~$i$ in this section.\todolater{verify whether this remains valid}
\todolater{gewichtung theoreme}

By definition of $\overrightarrow{\tau}$, $\delta$ naturally satisfies non-negativity and the triangle inequality.
In general, the directed distance $\overrightarrow{\tau}$ does not satisfy the identity of indiscernibles.
For instance, it holds that $\overrightarrow{\tau}(a\sim b,a\rhd b)=0$.
Nevertheless, for each $\unrhd$ over $N_i$ as divided into $\unrhd^+$ and $\unrhd^-$ we obtain two unique preference orders $C^+$ and $C^-$ which represent a unique coalition~$C\subseteq N_i\cup\{i\}$ (up to neutral players), such that $\delta(\unrhd,C)=0$. Indeed, $C$ is $i$'s favourite coalition (up to neutral players) and contains all of $i$'s friends and none of $i$'s enemies.\\[-3.5ex]
\begin{observation}\label{obs:favourite-coalition}
  The distance~$\delta(\unrhd,C)$ between a player's preference order~$\unrhd$ and a coalition $C$ is $0$ if and only if $N^+_i\subseteq C$ and $N^-\cap C=\emptyset$.
\end{observation}
%

Moreover, the comparability of two coalitions
is efficient
since we assume that for each player the number of known players is bounded by a constant. This conforms to the definition of a graphical hedonic game~\citep{pet:graphical-hg-bounded-tw:16}.
\todolater{positionen erwähnung graphical hg in gesamttext noch mal durchgehen}
In Proposition~\ref{prop:char-delta} we explicitly state
the calculation of the underlying distances.
\begin{observation}\label{obs:constant-time}
   The
  distance between a player's preference and a coalition can be computed in constant time.
\end{observation}


\todoomit{
The following basic properties hold for our model.
\begin{description}
    \item [Reflexivity:] The preference order $\succeq_i$ is \emph{reflexive} if  $A\succeq_i A$ for each coalition $A\in \mathcal{N}_i$.
    \item [Transitivity:] The preference order $\succeq_i$ is \emph{transitive} if  for any three coalitions $A,B,C \in \mathcal{N}_i$, $A \succeq_i B$ and $B \succeq_i C$ implies $A \succeq_i C$.
    \item [Anonymity:] The preference order $\succeq_i$ is \emph{anonymous} if renaming the players in $N$ does not change~$\succeq_i$.
\end{description}

Reflexivity follows immediately from the definition. The comparison between the coalitions is based on the $\geq$ relation for natural numbers which is reflexive. The same holds for transitivity.
As renaming the players does neither change the player's preference nor the allocation of players to a coalition and thus has no influence on the preference over coalitions, anonymity is fulfilled as well.
}%
Our model satisfies \emph{reflexivity} ($A\succeq A$ for each $A\in \mathcal{N}_i$) and \emph{transitivity} ($A \succeq B$ and $B \succeq C$ implies $A \succeq C$ for each $A,B,C \in \mathcal{N}_i$), since the comparison between the coalitions is based on the $\leq$ relation for natural numbers.
As renaming the players has no influence on the preference over coalitions, \emph{anonymity} is fulfilled by definition.
A further property, adapted from the definitions of citizen's sovereignty and nonimposition in the context of (committee) elections requests the possibility for each coalition to become a player $i$'s favourite coalition:
\todoomit{
\begin{description}
    \item [Nonimposition:] For a fixed player $i$ and each $C \in \mathcal{N}_i$ there exists a preference order $\unrhd_{i}$ such that $C$ ends up as $i$'s most preferred coalition.
\end{description}
}%
\emph{Nonimposition} holds if
for a player $i$ and each $C \in \mathcal{N}_i$ there exists a preference order $\unrhd$ such that $C$ ends up as $i$'s most preferred coalition.
Considering $\unrhd = \ballot{C_{\sim}}{i}{(N_{i} \setminus C)_{\sim}}$
for some $C\in \mathcal{N}_i$,
$C^+ 
= \ballot{C_{\sim}}{i}{}$ and
$C^-
= \ballot{}{i}{(N_{i} \setminus C)_{\sim}}$
equal $\unrhd^+_i$ and $\unrhd^-_i$, respectively.
Hence, $\delta(\unrhd,C)$ equals $0$.
%
\begin{proposition}
  \chgopds\ satisfy nonimposition.
\end{proposition}

\subsection{Changes within the Coalition}
Among the many options of how to extend a preference order over single players to an order over coalitions, there are some arguably reasonable basic rules for the comparison of two coalitions which only vary in the exchange or addition of one player (see, e.g.~\cite{bar-bos-pat:b:ranking-sets-of-objects}).
%
For instance, if $i$ considers
two players to be indifferent,
coalitions which differ only in these players are ranked equally by~$i$.
\todoomit{
\begin{description}
   \item [Indifference Between Players:] Let $j$ and 
   $k$ be two distinct players with $j \sim_{i} k$. Then, for all coalitions $C \in \mathcal{N}_i \setminus (\mathcal{N}_j \cup \mathcal{N}_k)$ it holds that 
   $C \cup \{j\} \succeq_i C \cup \{k\}$ and $C \cup \{k\} \succeq_i C \cup \{j\}$.
\end{description}
}%
\todoomit{
Let $j$ and 
   $k$ be two distinct players with $j \sim k$. Then,
   \emph{indifference between players} is satisfied if
   for all coalitions $C \in \mathcal{N}_i \setminus (\mathcal{N}_j \cup \mathcal{N}_k)$ it holds that 
   $C \cup \{j\} \succeq C \cup \{k\}$ and $C \cup \{k\} \succeq C \cup \{j\}$.
   %
\begin{proposition}
  \chgopds\ satisfy indifference between players.
\end{proposition}
}%
\begin{proposition}
  Let $i$, $j$, and 
   $k$ be
  players in a \hgopd\ with $j \sim_i k$. Then,
   for all coalitions $C \in \mathcal{N}_i \setminus (\mathcal{N}_j \cup \mathcal{N}_k)$, it holds that 
   $C \cup \{j\} \succeq_i C \cup \{k\}$ and $C \cup \{k\} \succeq_i C \cup \{j\}$.
\end{proposition}

In the same way, the idea requires that adding a friend to a coalition never downgrades this coalition while adding an enemy never
makes it more preferable.
%
\todoomit{
\begin{proposition}
\label{prop:add-friend}
  From a player $i$'s perspective, adding a friend to a coalition never makes this coalition less favourable while adding an enemy never makes it more preferable.
\end{proposition}
}%
%
%
\begin{proposition}
\label{prop:add-friend}
  From a player $i$'s perspective in a \hgopd, adding a friend to a coalition always improves this coalition
  while adding an enemy always makes it less favourable.
\end{proposition}
Formally, a player's
increasing satisfaction
with a coalition $A \in \mathcal{N}_i$ when adding a friend to this coalition is shown by proving that $\delta(\unrhd,A)
> \delta(\unrhd,A \cup \{x\})$ for $x \in N^+ \setminus A$.
Via the characterization in Equations~\ref{eq:char-deltaplus} and~\ref{eq:char-deltaminus} it suffices
to observe
that $-1 - \lvert \{b \in N^+ \mid x \rhd b\} \rvert
<
0$.
Similarly, to show that 
$\delta(\unrhd,A)
<\delta(\unrhd,A \cup \{y\})$ for $y \in N^- \setminus A$ and $A \in \mathcal{N}$ it suffices that $0
<1 + \lvert\{ f \in N^- \mid f \rhd y\} \rvert$.
\todoomit{
As both inequalities are even strict, adding a friend increases player's satisfaction with a coalition while adding an enemy decreases it.
}%
%
%
A weak variant of both findings is also implied by the
following
general notion.
%
%
%
%
 %
%
%
%
On the basis of these properties and an extension principle for ranked sets of objects by
\citet{bos-sch:j:minimal-paths-on-ordered-graphs},
\citet{lan-rey-rot-sch-sch:c:hgopt,ker-lan-rey-rot-sch-sch:j:hgopt}
define an extension principle for ordinal preference ballots with two thresholds.
The general idea is that more and better friends are preferred, while more and worse enemies are less preferred.
\begin{definition}[polarized responsive extension principle~\citep{ker-lan-rey-rot-sch-sch:j:hgopt}]
\label{def_bossong-schweigert}
Let $\unrhd$ be play\-er~$i$'s preference order over players $N_i$.
Moreover, let $A$ and $B$ be two coalitions in~$\mathcal{N}_i$. 
The partial extension principle to preferences $\succeq^{+0-}$ over coalitions is defined by 
\begin{align*}
 A\succeq^{+0-} B\iff & \text{there exist two injective functions}\\
 &\phi: B\cap N^+ \hookrightarrow A\cap N^+
 \text{ with } \phi(j)\unrhd j \text{ for } j\in B\cap N^+
   \text{ and }\\
 &\psi: A\cap N^- \hookrightarrow B\cap N^-
 \text{ with } k\unrhd \psi(k) \text{ for } k\in A\cap N^-.
\end{align*}
\end{definition}\label{def:bs-principle}
We show that our model is compatible with this extension principle. That is,
if two coalitions $A,B\in \mathcal{N}_i$ satisfy $A\succeq^{+0-} B$, it also holds that $A\succeq B$. Note that the reverse implication is not required here, since $\succeq^{+0-}$ allows indecisions between coalitions.
\begin{theorem}\label{thm:bs-compatibility}
The preferences $\succeq$ of a player over coalitions in a \hgopd\ are compatible with $\succeq^{+0-}$.
\end{theorem}
\begin{proof}
Let $A\succeq^{+0-} B$ and let $\phi$ and $\psi$ be the two injective functions as in Definition~\ref{def:bs-principle}. We want to show that $A\succeq B$, i.e., $\delta(\unrhd,B)\geq\delta(\unrhd,A)$.
Firstly, by Equation~\eqref{eq:char-deltaminus} for the players in $N^-$ it holds that
  $\delta^-(\unrhd,B)-\delta^-(\unrhd,A)
  =\lvert B\cap N^-\rvert + \sum_{f\in N^-} \left\lvert \left\{b\in B\cap N^- \mid f\rhd b \right\} \right\rvert
  - \lvert A\cap N^-\rvert - \sum_{f\in N^-} \left\lvert \left\{a\in A\cap N^- \mid f\rhd a \right\} \right\rvert$.
Since $\psi$ is injective, $B\cap N^-$ contains at least as many players as $A\cap N$.
Therefore,\\[-3.5ex]
\begin{align}\lvert B\cap N^-\rvert - \lvert A\cap N^-\rvert\geq 0.\label{eq:bs-minus-part1}\end{align}
Furthermore, for each $f\in N^-$ it holds that if some $a\in A\cap N^-$ satisfies $f\rhd a$, then $f\rhd \psi(a)$.
Hence,
\begin{align*}
  &\left\lvert \left\{b\in B\cap N^- \mid f\rhd b \right\} \right\rvert
  - \left\lvert \left\{a\in A\cap N^- \mid f\rhd a \right\} \right\rvert
   \\[1ex]
  &\qquad\quad= 
  \underbrace{\left\lvert \left\{b\in B\cap N^- \mid f\rhd b, b\notin\mathrm{image}(\psi) \right\} \right\rvert}_{\geq 0}
   + \underbrace{\left\lvert \left\{b\in \mathrm{image}(\psi) \mid f\rhd b \right\} \right\rvert}_{=\left\lvert\{a\in A\cap N^-\mid f\rhd\psi(a)\}\right\rvert}
  \nonumber  
  \\
  &\qquad\quad\quad
   - \left\lvert \left\{a\in A\cap N^- \mid f\rhd a \right\} \right\rvert
  \nonumber
  \end{align*}
\begin{align}
  &\qquad\quad\geq \left\lvert\{a\in A\cap N^-\mid f\rhd\psi(a)\}\right\rvert
   - \left\lvert \left\{a\in A\cap N^- \mid f\rhd a \right\} \right\rvert
   \geq 0.
\label{eq:bs-minus-part2}
\end{align}
By \eqref{eq:bs-minus-part1} and summing up \eqref{eq:bs-minus-part2} over each $f\in N^-$,
we obtain
$\delta^-(\unrhd,B)-\delta^-(\unrhd,A)\geq 0$.
%
Secondly, for the players in $N^+$, Equation~\eqref{eq:char-deltaplus}
% ...
can be used to argue that $\delta^+(\unrhd,B)-\delta^+(\unrhd,A)\geq 0$.
Since $\lvert A\cap N^+\rvert \geq \lvert B\cap N^+\rvert$, $\lvert N^+\setminus B\rvert \geq \lvert N^+\setminus A\rvert$ is implied.
The sum over $f\in N^+\setminus B$ can be divided into the sum over all $f\in N^+$ minus the sum over $f\in B\cap N^+$, such that a similar argument as for $\delta^-$ can be applied.
All in all, it holds that
  $\delta(\unrhd,B)-\delta(\unrhd,A)=\delta^+(\unrhd,B)-\delta^+(\unrhd,A)+\delta^-(\unrhd,B)-\delta^-(\unrhd,A)\geq 0$.
\end{proof}

\subsection{Changes within the Preference}
If a player's position is changed within a preference order, it is common to assume certain monotonicity properties (see, e.g.~\cite{bra-con-end-lan-pro:b:handbook-comsoc}).
%
The following relaxation of monotonicity holds: If a player~$j$'s position is improved in a preference ranking, a coalition not containing $j$ can only outperform a previously preferred coalition containing $j$ by the number of swaps during the improvement.
\begin{theorem}
Let $\unrhd$ be the preference order of player~$i$ over the set of players in a \hgopd, $j,x$ two players in $N_i$, and $A,B\in\mathcal{N}_i$ two coalitions with $j\in A$ and $j\notin B$ and $A\succeq B$.
If player $i$ changes their preference $\unrhd$ to $\unrhd'$ by shifting $j$ to a better position than $x$, it holds that
$
  \delta^+(\unrhd,B)-\delta^+(\unrhd,A)\geq -\tau^*(\unrhd,\unrhd')
\text{ and }
  \delta^-(\unrhd,B)-\delta^-(\unrhd,A)\geq -\tau^*(\unrhd,\unrhd').
$
\end{theorem}
\todoomit{
In fact, we obtain \dots...
$\geq -\lvert\{u\in N^+\setminus A\mid u\sim' j\}\rvert-\lvert\{y\in N^-\cap B\mid y\sim j\}\rvert$
}%
Note that we do not obtain
a difference of at most $0$.
This relates to a classical notion of monotonicity in which a player $j$ is improved within player $i$'s preference. At the same time it is crucial that this improvement is detached from any changes regarding other players. In our setting, however, player $j$'s improvement is immediately associated with a deterioration of at least one other player. Thus, that notion of monotonicity is not applicable.



\subsection{Distinction from other Hedonic Games}

The notions of $\delta^+$ and $\delta^-$ allow a certain degree of flexibility in the expressivity of \hgopds. In its current form, $\delta$ being the $1$-norm of $\delta^+$ and $\delta^-$ neither specifies a tendency towards friend appreciation nor enemy aversion, but a combination of both. Either edge case can be expressed by multiplying $\delta^+$ or $\delta^-$ with an approriate weight.

\chgopds\ do satisfy \emph{additive separability}.
In fact,
equivalent preferences can be encoded by additive utilities $u_i(a)=\delta^+(\unrhd^+_i,N^+_i\setminus\{b\})$ and $u_i(b)=\delta^-(\unrhd^-_i,\{b\})$ for players $i$, $a\in N^+_i$, and $b\in N^-_i$.
If $\delta$ is altered to, e.g., the $2$-norm, this no longer holds. Then, we can express independent coalition relations
such as $\{a,b,c,d,i\}\succ_i\{a,i\}$ for $\unrhd_i$ as in Example~\ref{ex:running-example-1},
but $\{a,e,i\}\succ_i\{a,b,c,d,e,i\}$.

The model distinguishes from other known extensions of ordinal player preferences.
For instance, there exists games for this model (e.g., the relation $\{i,b\}\succ_i\{a,b,c,i\}$ via $a\rhd i\rhd b\rhd c$) which cannot be expressed by $\mathcal{B}$- or $\mathcal{W}$-preferences~\citep{cec-rom:j:singleton-encoding}.
However, there also exist $\mathcal{B}$-preferences ($\{a,i\}\succ_i\{a,b,i\}\succ_i\{b,i\}\succ_i\{i\}$) and $\mathcal{W}$-preferences ($\{i\}\succ_i\{c,i\}\succ_i\{d,i\}$ and indifference between $\{d,i\}$ and $\succeq_i\{i,c,d\}$) which cannot be expressed here.

\section{Stability}
\label{ch_stability}

A frequent question for hedonic games is whether a coalition structure is stable in some sense.
The best outcome would be a perfect coalition structure, where each player is in their favourite coalition. We obtain Proposition~\ref{prop:perfect} for the verification and the existence problem for perfection.
By Observation~\ref{obs:favourite-coalition},
we know that each player's favourite coalition~$F_i$ satisfies $\delta(\unrhd_i,F_i)=0$. Hence, in order to decide whether a given coalition structure is perfect, it has to be verified whether $\delta(\unrhd_i,\Gamma(i))=0$ for each $i\in N$ which by Observation~\ref{obs:constant-time} can be determined in constant time.
For the existence problem, the proof is constructive: Similar to a breadth first search, players are added consecutively to coalitions until they either have a conflict or form a perfect coalition structure. Since for each player there are at most $d$ known players, the search runs in $\mathcal{O}(n+d\cdot n)=\mathcal{O}(n)$ time.

\begin{proposition}\label{prop:perfect}
It can be verified in $\mathcal{O}(n)$ time whether a given coalition structure in a \hgopd\ is perfect.
Furthermore, it can be decided in $\mathcal{O}(n)$ time whether a given \hgopd\ allows a perfect coalition structure.
\end{proposition}
%
%
%


A more likely stable outcome is that of a coalition structure no player has an incentive to deviate from.
Example~\ref{ex:not-nash-stable} shows that a Nash-stable and an individually stable coalition structure do not always exist. However, it can be shown
that a contractually individually stable coalition structure always exists. It can be seen that when a player wants to deviate and is welcome in the new coalition and not bound to the former coalition the sum of all players' distances decreases which can reach a minimum.
\begin{example}\label{ex:not-nash-stable}
Consider a game consisting of five players $a$, $b$, $c$, $d$, and $e$
with $\unrhd_a=\ballot{b}{a}{c\sim_a d}$, $\unrhd_b=\ballot{c}{b}{d\sim_b e}$, etc.\ continued rotationally symmetricly.
Then for each coalition structure, it can be seen that there is always at least one player who rather wants to play alone or is welcome joining a friend. Thus, no outcome is individually stable and consequently also not Nash-stable.
\end{example}

The verification problem for these stability notions can be decided in polynomial time.
\todoomit{
Due to the degree bound~$d$, there are at most $d+1$ possible coalitions a player~$i$ might want to deviate to, namely those containing a friend of $i$'s.
}%
For (contractual) individual stability,
\todoomit{
it additionally needs to be verified whether the members of the new coalition welcome $i$ (and whether $i$ is not bound to the current coalition).
There is only an impact on those players who know $i$. Therefore,
}%
we need to distinguish between graphs which also have a bounded in-degree (players can only be known by a limited number of players and followers are encoded in the graph) and those which allow an unbounded in-degree (players can have many followers they do not know themselves).
\begin{theorem}\label{thm:Nash-verification}
  Given a \hgopd\ and a coalition structure~$\Gamma$, it can be verified whether $\Gamma$ is\\[-4.5ex]
  \begin{itemize}
   \item Nash-stable in time in $\mathcal{O}(n)$,
   \item (contractually) individually stable in time in $\mathcal{O}(n)$ for bounded in-degree,
   \item (contractually) individually stable in time in $\mathcal{O}(n^2)$ for unbounded in-degree.
  \end{itemize}
\end{theorem}
\begin{proof}
  Due to the degree bound~$d$, there are at most $d+1$ possible coalitions a player~$i$ might want to deviate to, namely those containing a friend of $i$'s, $\Gamma(j)$ for each $j\in N^+_i$, or the singleton coalition. For each player, we only need to find out whether $i$ prefers to play in one of those candidate coalitions to $i$'s current coalition. This can be determined in time independent from~$n$ by Observation~\ref{obs:constant-time}. If one player prefers one candidate coalition, $\Gamma$ is not Nash-stable.
  
  For individual stability, it additionally needs to be verified whether the members of the new coalition welcome $i$. There is only an impact on those players who know $i$. 
  For players~$j$ with $i\in N^+_j$, $i$ is welcome (see Proposition~\ref{prop:add-friend}); for those with $i\in N^-_j$, $i$ is not welcome.
  We need to check whether there are any of the latter players in a candidate coalition.
  If $i$'s followers are known and the in-degree is also bounded, this can easily be verified independently from~$n$.
  If the incoming edges are not encoded or the in-degree is unbounded, in the worst case we have to ask every player in the candidate coalition whether they don't like $i$. This requires an additional factor $n$ for the running time.
  The coalition structure $\Gamma$ is not individually stable, if one player~$i$ prefers one candidate coalition which does not contain any players who consider $i$ an enemy.
  
  Similarly, for contractual individual stability, it additionally needs to be verified whether $i$ is not bound to the current coalition. The same distiction for incoming edges needs to be made.   
  If one player~$i$ prefers one candidate coalition which does not contain any players who consider $i$ an enemy, and $\Gamma(i)$ does not contain any players who consider
  $i$ a friend, $\Gamma$ is contractually individually stable.
\end{proof}

Since \hgopds\ are graphical hedonic games,
it follows from a result by
\citet[Theorem~5]{pet:graphical-hg-bounded-tw:16} that if additionally the treewidth of the underlying dependency graph is bounded by a constant, it can be decided in time linear in the number of players whether a Nash-stable coalition structure exists.
In general, this problem is NP-complete. The NP-hardness reduction is inspired by another result by 
\citet[Theorem~8]{pet:graphical-hg-bounded-tw:16}, adapted to our model and a degree bound of $6$, but an unbounded treewidth.
\todoomit{
For an instance of \textsc{Exact Cover by Three Sets}, we may assume that each element $x_1,\dots,x_n$ occurs in at most three sets $S_1,\dots,S_m$.
We create a game with $2n+4m$ players ($a_i$, $b_i$ for each $x_i$, and $s_j$, $t_{j_1}$, $t_{j_2}$, $t_{j_3}$ for each $S_j$)
with the following preferences
for $a_i$, $b_i$, $s_j$, and $t_{j_k}$, respectively:
$\ballot{b_i}{a_i}{\{s_j\mid x_i\in S_j\}_\sim}$,
$\ballot{ }{b_i}{ \{a_i,t_{j_k}\mid x_i \text{ is the $k$th element in }S_j\}_\sim }$,
$\ballot{\{b_i,t_{j_1},t_{j_2},t_{j_3}\mid x_i\in S_j\}_\sim}{s_j}{ }$,\linebreak and
$\ballot{\{t_{j_\ell}\mid 1\leq\ell\leq 3, \ell\neq k\}_\sim}{t_{j_k}}{ }$.
%
The original instance allows an exact cover if and only if there is a Nash-stable coalition structure in the constructed game.
}%
Note that the same proof holds for $\delta$ being the $2$-norm instead of the $1$-norm.
\begin{theorem}
 Given a \hgopd, it is NP-complete to decide whether a Nash-stable coalition structure exists.
\end{theorem}
\begin{proof}
 The problem is contained in NP, since, for a chosen coalition structure, it can be verified whether it is Nash-stable in
 polynomial time in the number of players
 by Theorem~\ref{thm:Nash-verification}.
 
 For the lower bound, we consider a reduction from \textsc{Exact Cover by Three Sets}.
 We may assume that each element $x_1,\dots,x_n$ of an instance occurs in at most three sets $S_1,\dots,S_m$.
 Given such an instance, 
 we create a game with $2n+4m$ players ($a_i$, $b_i$ for each $x_i$, and $s_j$, $t_{j_1}$, $t_{j_2}$, $t_{j_3}$ for each $S_j$)
 in linear time
 with the following preferences
 for $a_i$, $b_i$, $s_j$, and $t_{j_k}$, respectively
 (see also Figure~\ref{fig:reduction-nash-stable-existence}).
\begin{figure}[ht]
    \centering
    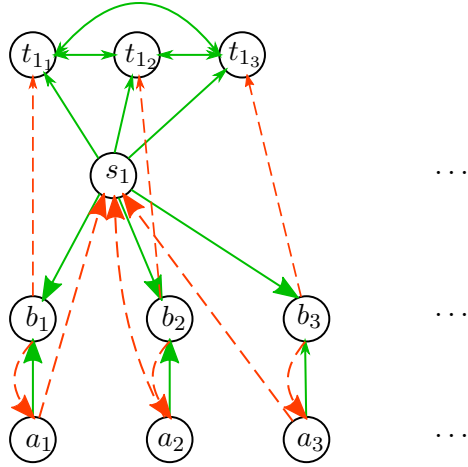
    \label{fig:reduction-nash-stable-existence}
    \caption{Polynomial reduction from \textsc{Exact Cover by Three Sets} to Nash stability existence: Gadget for one example set $S_1=\{x_1,x_2,x_3\}$. Solid (green) edges indicate a friendship relation, dashed (red) edges indicate dislike}
\end{figure}
 \begin{itemize}
  \item For each $x_i$ there are two players $a_i$ and $b_i$,
  \item for each $S_j$ there are four players $s_j$ and $t_{j_1}$, $t_{j_2}$ and $t_{j_3}$,
  \item $a_i$'s preferences are $\ballot{b_i}{a_i}{\{s_j\mid x_i\in S_j\}_\sim}$,
  \item $b_i$'s preferences are $\ballot{ }{b_i}{ \{a_i,t_{j_k}\mid x_i \text{ is the $k$th element in }S_j\}_\sim }$,
  \item $s_j$'s preferences are $\ballot{\{b_i,t_{j_1},t_{j_2},t_{j_3}\mid x_i\in S_j\}_\sim}{s_j}{ }$, and
  \item $t_{j_k}$'s preferences are $\ballot{\{t_{j_\ell}\mid 1\leq\ell\leq 3, \ell\neq k\}_\sim}{t_{j_k}}{ }$.
 \end{itemize}
 It holds that the original instance allows an exact cover if and only if there exists a Nash-stable coalition structure in the constructed game.
 If, on the one hand, an exact cover $C\subseteq\{S_1,\dots,S_m\}$ exists, the coalition structure
 $
  \{\{a_i\}\mid 1\leq i\leq n\}\cup\{\{s_j,b_{j_1},b_{j_2},b_{j_3}\mid S_j=\{x_{j_1},x_{j_2},x_{j_3}\}\},\{t_{j_1},t_{j_2},t_{j_3}\}\mid S_j\in C\}
  \cup \{\{s_j,t_{j_1},t_{j_2},t_{j_3}\}\}\mid S_j\notin C\}
 $
 is Nash-stable.
 In fact, no $a_i$ has an incentive to move, since they are indifferent between playing alone and playing with $b_i$, but also with a corresponding $s_j$.
 Each $b_i$ is as happy as possible, since they are not playing together with $a_i$ or any $t_{j_k}$.
 Moreover, each $s_j$ cannot improve, since $s_j$ is indifferent between $\{s_j,b_{j_1},b_{j_2},b_{j_3}\mid S_j=\{x_{j_1},x_{j_2},x_{j_3}\}$ and $\{s_j,t_{j_1},t_{j_2},t_{j_3}\}$.
 The players $t_{j_k}$
 don't want to move, as they are in one of their favourite coalitions.
 
 If, on the other hand, a stable coalition structure exists, we can observe the following: For each $j$, $1,\leq j\leq m$, all three $t_{j_k}$, $1\leq k\leq 3$ have to play in one coalition, otherwise there exists one $t_{j_k}$ that has an incentive to move to another player of that group.
 Each $b_i$, $1,\leq i\leq n$, cannot play together with $a_i$ or $t_{j_k}$, otherwise $b_i$ would move to the singleton coalition.
 Therefore, $s_j$ can only be in a coalition with at most three friends, either the group of $t_{j_k}$ players or corresponding $b_i$ players.
 If the latter are less than three $b_i$ players, $s_j$ moves to the former group.
 Thus, $s_j$'s coalition contains exactly three friends.
 If some $b_i$ does not play together with one corresponding $s_j$, then $a_i$ would move to $b_i$'s coalition.
 Hence, there exists an exact cover, namely by those $S_j$ for which $s_j$ plays together with those $b_i$ with $x_i\in S_j$.
\end{proof}

\section{Conclusion and Future Work}
We have developed a model of hedonic games based on distances between players ordinal preferences and a suitable representation of coalitions.
An axiomatic and computational study indicates that this model is a reasonable completion of compact hedonic game representations.
For future work,
it is interesting to study this model with respect to further stability concepts, such as the core.
Furthermore, we are interested in a comparison to a different preference model, e.g., relaxing the constant bound on known players.
Additionally, we can extend
the model
to other distance norms or adapt further ordinal distance measures.
%
%
%
%
Moreover, the distance-based approach enables a comparison
between
games, in scenarios which require a distance between games.

\printbibliography




\end{document}

%% file: abb-char-deltaplus.pdf_tex
\begingroup%
  \makeatletter%
  \providecommand\color[2][]{%
    \errmessage{(Inkscape) Color is used for the text in Inkscape, but the package 'color.sty' is not loaded}%
    \renewcommand\color[2][]{}%
  }%
  \providecommand\transparent[1]{%
    \errmessage{(Inkscape) Transparency is used (non-zero) for the text in Inkscape, but the package 'transparent.sty' is not loaded}%
    \renewcommand\transparent[1]{}%
  }%
  \providecommand\rotatebox[2]{#2}%
  \newcommand*\fsize{\dimexpr\f@size pt\relax}%
  \newcommand*\lineheight[1]{\fontsize{\fsize}{#1\fsize}\selectfont}%
  \ifx\svgwidth\undefined%
    \setlength{\unitlength}{394.8799998bp}%
    \ifx\svgscale\undefined%
      \relax%
    \else%
      \setlength{\unitlength}{\unitlength * \real{\svgscale}}%
    \fi%
  \else%
    \setlength{\unitlength}{\svgwidth}%
  \fi%
  \global\let\svgwidth\undefined%
  \global\let\svgscale\undefined%
  \makeatother%
  \begin{picture}(1,0.41879959)%
    \lineheight{1}%
    \setlength\tabcolsep{0pt}%
    \put(0,0){\includegraphics[width=\unitlength,page=1]{abb-char-deltaplus.pdf}}%
    \put(0.15574356,0.33522989){\color[rgb]{0,0,0}\makebox(0,0)[lt]{\lineheight{1.25}\smash{\begin{tabular}[t]{l}$1$    \end{tabular}}}}%
    \put(0.4786265,0.33902852){\color[rgb]{0,0,0}\makebox(0,0)[lt]{\lineheight{1.25}\smash{\begin{tabular}[t]{l}$q$\end{tabular}}}}%
    \put(0,0){\includegraphics[width=\unitlength,page=2]{abb-char-deltaplus.pdf}}%
    \put(0.10256277,0.38081336){\color[rgb]{0,0,0}\makebox(0,0)[lt]{\lineheight{1.25}\smash{\begin{tabular}[t]{l}$\unrhd_i(C\cap N_i^+)$\end{tabular}}}}%
    \put(0.36466777,0.38081337){\color[rgb]{0,0,0}\makebox(0,0)[lt]{\lineheight{1.25}\smash{\begin{tabular}[t]{l}$i$\end{tabular}}}}%
    \put(0.40645249,0.38081336){\color[rgb]{0,0,0}\makebox(0,0)[lt]{\lineheight{1.25}\smash{\begin{tabular}[t]{l}$\overleftarrow{\unrhd_i}(N_i^+\setminus C)$\end{tabular}}}}%
    \put(0.05697934,0.26685469){\color[rgb]{0,0,0}\makebox(0,0)[lt]{\lineheight{1.25}\smash{\begin{tabular}[t]{l}$1$\end{tabular}}}}%
    \put(0.05697934,0.19468088){\color[rgb]{0,0,0}\makebox(0,0)[lt]{\lineheight{1.25}\smash{\begin{tabular}[t]{l}$p$\end{tabular}}}}%
    \put(0,0){\includegraphics[width=\unitlength,page=3]{abb-char-deltaplus.pdf}}%
    \put(0.36086914,0.01614559){\color[rgb]{0,0,0}\makebox(0,0)[lt]{\lineheight{1.25}\smash{\begin{tabular}[t]{l}$1$\end{tabular}}}}%
    \put(0.05697934,0.01994422){\color[rgb]{0,0,0}\makebox(0,0)[lt]{\lineheight{1.25}\smash{\begin{tabular}[t]{l}$r_i$\end{tabular}}}}%
    \put(0.02343463,0.13893069){\color[rgb]{0,0,0}\rotatebox{90}{\makebox(0,0)[lt]{\lineheight{1.25}\smash{\begin{tabular}[t]{l}$\unrhd_i(N^+)$\end{tabular}}}}}%
    \put(-0,0.01994422){\color[rgb]{0,0,0}\makebox(0,0)[lt]{\lineheight{1.25}\smash{\begin{tabular}[t]{l}$i$\end{tabular}}}}%
    \put(0,0){\includegraphics[width=\unitlength,page=4]{abb-char-deltaplus.pdf}}%
    \put(0.55080023,0.33902852){\makebox(0,0)[lt]{\lineheight{1.25}\smash{\begin{tabular}[t]{l}$r_C$ \end{tabular}}}}%
    \put(-0,0.38081336){\makebox(0,0)[lt]{\lineheight{1.25}\smash{\begin{tabular}[t]{l}$\delta^+$\end{tabular}}}}%
    \put(0.5345797,0.17340353){\makebox(0,0)[lt]{\lineheight{1.25}\smash{\begin{tabular}[t]{l}$\textcolor{gray}f$\end{tabular}}}}%
    \put(0,0){\includegraphics[width=\unitlength,page=5]{abb-char-deltaplus.pdf}}%
  \end{picture}%
\endgroup%

%% file: abb-char-deltaminus.pdf_tex
\begingroup%
  \makeatletter%
  \providecommand\color[2][]{%
    \errmessage{(Inkscape) Color is used for the text in Inkscape, but the package 'color.sty' is not loaded}%
    \renewcommand\color[2][]{}%
  }%
  \providecommand\transparent[1]{%
    \errmessage{(Inkscape) Transparency is used (non-zero) for the text in Inkscape, but the package 'transparent.sty' is not loaded}%
    \renewcommand\transparent[1]{}%
  }%
  \providecommand\rotatebox[2]{#2}%
  \newcommand*\fsize{\dimexpr\f@size pt\relax}%
  \newcommand*\lineheight[1]{\fontsize{\fsize}{#1\fsize}\selectfont}%
  \ifx\svgwidth\undefined%
    \setlength{\unitlength}{309.51746891bp}%
    \ifx\svgscale\undefined%
      \relax%
    \else%
      \setlength{\unitlength}{\unitlength * \real{\svgscale}}%
    \fi%
  \else%
    \setlength{\unitlength}{\svgwidth}%
  \fi%
  \global\let\svgwidth\undefined%
  \global\let\svgscale\undefined%
  \makeatother%
  \begin{picture}(1,0.53430129)%
    \lineheight{1}%
    \setlength\tabcolsep{0pt}%
    \put(0,0){\includegraphics[width=\unitlength,page=1]{abb-char-deltaminus.pdf}}%
    \put(0.19869643,0.42768371){\color[rgb]{0,0,0}\makebox(0,0)[lt]{\lineheight{1.25}\smash{\begin{tabular}[t]{l}$1$    \end{tabular}}}}%
    \put(0.61062799,0.43252997){\color[rgb]{0,0,0}\makebox(0,0)[lt]{\lineheight{1.25}\smash{\begin{tabular}[t]{l}$q$\end{tabular}}}}%
    \put(0,0){\includegraphics[width=\unitlength,page=2]{abb-char-deltaminus.pdf}}%
    \put(0.13084879,0.48583875){\color[rgb]{0,0,0}\makebox(0,0)[lt]{\lineheight{1.25}\smash{\begin{tabular}[t]{l}$\overleftarrow{\unrhd_i}(C\cap N_i^-)$\end{tabular}}}}%
    \put(0.46524033,0.48583876){\color[rgb]{0,0,0}\makebox(0,0)[lt]{\lineheight{1.25}\smash{\begin{tabular}[t]{l}$i$\end{tabular}}}}%
    \put(0.51854895,0.48583875){\color[rgb]{0,0,0}\makebox(0,0)[lt]{\lineheight{1.25}\smash{\begin{tabular}[t]{l}$\unrhd_i(N_i^-\setminus C)$\end{tabular}}}}%
    \put(0.0726938,0.36468245){\color[rgb]{0,0,0}\makebox(0,0)[lt]{\lineheight{1.25}\smash{\begin{tabular}[t]{l}$1$\end{tabular}}}}%
    \put(0.0726938,0.1756785){\color[rgb]{0,0,0}\makebox(0,0)[lt]{\lineheight{1.25}\smash{\begin{tabular}[t]{l}$p$\end{tabular}}}}%
    \put(0,0){\includegraphics[width=\unitlength,page=3]{abb-char-deltaminus.pdf}}%
    \put(0.46039407,0.36468273){\color[rgb]{0,0,0}\makebox(0,0)[lt]{\lineheight{1.25}\smash{\begin{tabular}[t]{l}$1$\end{tabular}}}}%
    \put(0.0726938,0.02544468){\color[rgb]{0,0,0}\makebox(0,0)[lt]{\lineheight{1.25}\smash{\begin{tabular}[t]{l}$r_i$\end{tabular}}}}%
    \put(0.02989772,0.10455284){\color[rgb]{0,0,0}\rotatebox{90}{\makebox(0,0)[lt]{\lineheight{1.25}\smash{\begin{tabular}[t]{l}$\unrhd_i(N^-)$\end{tabular}}}}}%
    \put(0.00000001,0.36468273){\color[rgb]{0,0,0}\makebox(0,0)[lt]{\lineheight{1.25}\smash{\begin{tabular}[t]{l}$i$\end{tabular}}}}%
    \put(0,0){\includegraphics[width=\unitlength,page=4]{abb-char-deltaminus.pdf}}%
    \put(0.70270668,0.43252997){\makebox(0,0)[lt]{\lineheight{1.25}\smash{\begin{tabular}[t]{l}$r_C$ \end{tabular}}}}%
    \put(0,0.48583875){\makebox(0,0)[lt]{\lineheight{1.25}\smash{\begin{tabular}[t]{l}$\delta^-$\end{tabular}}}}%
    \put(0.68201266,0.22122689){\makebox(0,0)[lt]{\lineheight{1.25}\smash{\begin{tabular}[t]{l}$\textcolor{gray}f$\end{tabular}}}}%
    \put(0,0){\includegraphics[width=\unitlength,page=5]{abb-char-deltaminus.pdf}}%
    \put(0.22646489,0.21075439){\makebox(0,0)[lt]{\lineheight{1.25}\smash{\begin{tabular}[t]{l}$\textcolor{gray}f$\end{tabular}}}}%
    \put(0,0){\includegraphics[width=\unitlength,page=6]{abb-char-deltaminus.pdf}}%
  \end{picture}%
\endgroup%

%% file: abb_reduktion.pdf_tex
\begingroup%
  \makeatletter%
  \providecommand\color[2][]{%
    \errmessage{(Inkscape) Color is used for the text in Inkscape, but the package 'color.sty' is not loaded}%
    \renewcommand\color[2][]{}%
  }%
  \providecommand\transparent[1]{%
    \errmessage{(Inkscape) Transparency is used (non-zero) for the text in Inkscape, but the package 'transparent.sty' is not loaded}%
    \renewcommand\transparent[1]{}%
  }%
  \providecommand\rotatebox[2]{#2}%
  \newcommand*\fsize{\dimexpr\f@size pt\relax}%
  \newcommand*\lineheight[1]{\fontsize{\fsize}{#1\fsize}\selectfont}%
  \ifx\svgwidth\undefined%
    \setlength{\unitlength}{202.91157001bp}%
    \ifx\svgscale\undefined%
      \relax%
    \else%
      \setlength{\unitlength}{\unitlength * \real{\svgscale}}%
    \fi%
  \else%
    \setlength{\unitlength}{\svgwidth}%
  \fi%
  \global\let\svgwidth\undefined%
  \global\let\svgscale\undefined%
  \makeatother%
  \begin{picture}(1,0.85416082)%
    \lineheight{1}%
    \setlength\tabcolsep{0pt}%
    \put(0,0){\includegraphics[width=\unitlength,page=1]{abb_reduktion.pdf}}%
    \put(0.03006945,0.02982757){\color[rgb]{0,0,0}\makebox(0,0)[lt]{\lineheight{1.25}\smash{\begin{tabular}[t]{l}$a_1$\\\end{tabular}}}}%
    \put(0.2727042,0.03139645){\color[rgb]{0,0,0}\makebox(0,0)[lt]{\lineheight{1.25}\smash{\begin{tabular}[t]{l}$a_2$\end{tabular}}}}%
    \put(0.52877581,0.03115269){\color[rgb]{0,0,0}\makebox(0,0)[lt]{\lineheight{1.25}\smash{\begin{tabular}[t]{l}$a_3$\end{tabular}}}}%
    \put(0,0){\includegraphics[width=\unitlength,page=2]{abb_reduktion.pdf}}%
    \put(0.03006945,0.25333894){\color[rgb]{0,0,0}\makebox(0,0)[lt]{\lineheight{1.25}\smash{\begin{tabular}[t]{l}$b_1$\\\end{tabular}}}}%
    \put(0.28009659,0.25490782){\color[rgb]{0,0,0}\makebox(0,0)[lt]{\lineheight{1.25}\smash{\begin{tabular}[t]{l}$b_2$\end{tabular}}}}%
    \put(0.52877581,0.25786924){\color[rgb]{0,0,0}\makebox(0,0)[lt]{\lineheight{1.25}\smash{\begin{tabular}[t]{l}$b_3$\end{tabular}}}}%
    \put(0,0){\includegraphics[width=\unitlength,page=3]{abb_reduktion.pdf}}%
    \put(0.17660313,0.52772286){\color[rgb]{0,0,0}\makebox(0,0)[lt]{\lineheight{1.25}\smash{\begin{tabular}[t]{l}$s_1$\end{tabular}}}}%
    \put(0,0){\includegraphics[width=\unitlength,page=4]{abb_reduktion.pdf}}%
    \put(0.02267706,0.74227952){\color[rgb]{0,0,0}\makebox(0,0)[lt]{\lineheight{1.25}\smash{\begin{tabular}[t]{l}$t_{1_1}$\\\end{tabular}}}}%
    \put(0.21268967,0.7438484){\color[rgb]{0,0,0}\makebox(0,0)[lt]{\lineheight{1.25}\smash{\begin{tabular}[t]{l}$t_{1_2}$\end{tabular}}}}%
    \put(0.40310518,0.74680982){\color[rgb]{0,0,0}\makebox(0,0)[lt]{\lineheight{1.25}\smash{\begin{tabular}[t]{l}$t_{1_3}$\end{tabular}}}}%
    \put(0,0){\includegraphics[width=\unitlength,page=5]{abb_reduktion.pdf}}%
    \put(0.78189602,0.2726496){\color[rgb]{0,0,0}\makebox(0,0)[lt]{\lineheight{1.25}\smash{\begin{tabular}[t]{l}$\dots$\\\end{tabular}}}}%
    \put(0.78189602,0.04913183){\color[rgb]{0,0,0}\makebox(0,0)[lt]{\lineheight{1.25}\smash{\begin{tabular}[t]{l}$\dots$\\\end{tabular}}}}%
    \put(0.78189229,0.53511525){\color[rgb]{0,0,0}\makebox(0,0)[lt]{\lineheight{1.25}\smash{\begin{tabular}[t]{l}$\dots$\\\end{tabular}}}}%
    \put(0,0){\includegraphics[width=\unitlength,page=6]{abb_reduktion.pdf}}%
  \end{picture}%
\endgroup%